\documentclass[11pt]{article}
\usepackage[margin=1in]{geometry}
\usepackage{amsmath,amsfonts,amssymb,amsthm}
\usepackage{graphicx}
\usepackage{enumerate}
\usepackage{bbm}
\usepackage{verbatim}
\usepackage{hyperref,color}
\usepackage[capitalize,nameinlink]{cleveref}
\usepackage[dvipsnames]{xcolor}
\hypersetup{
	colorlinks=true,
	pdfpagemode=UseNone,
	citecolor=OliveGreen,
	linkcolor=NavyBlue,
	urlcolor=Magenta,
	pdfstartview=FitW
}
\usepackage{appendix}
\usepackage{makecell}
\usepackage{tablefootnote}
\crefname{appsec}{Appendix}{Appendices}
\usepackage{tikz}
\usepackage{pgfplots}
\usepackage{xifthen}
\usepackage{subcaption}
\usepackage{hhline}

\usepackage{algorithm}
\usepackage{algpseudocode}
\algnewcommand\algorithmicinput{\textbf{Input:}}
\algnewcommand\algorithmicoutput{\textbf{Initialization:}}
\algnewcommand\Input{\item[\algorithmicinput]}
\algnewcommand\Output{\item[\algorithmicoutput]}
\usepackage{array}
\usepackage{caption}

\theoremstyle{plain}
\newtheorem{theorem}{Theorem}[section]
\newtheorem{proposition}[theorem]{Proposition}
\newtheorem{lemma}[theorem]{Lemma}

\newtheorem{fact}[theorem]{Fact}

\theoremstyle{definition}
\newtheorem{definition}[theorem]{Definition}

\newtheorem*{assumption*}{Assumption}

\theoremstyle{remark}
\newtheorem{remark}[theorem]{Remark}

\crefname{lemma}{Lemma}{Lemmas}
\crefname{theorem}{Theorem}{Theorems}
\crefname{definition}{Definition}{Definitions}
\crefname{fact}{Fact}{Facts}
\crefname{claim}{Claim}{Claims}
\crefname{proposition}{Proposition}{Propositions}

\newcommand{\norm}[1]{\left\lVert #1 \right\rVert}
\newcommand{\TV}[2]{d_{\mathrm{TV}}\left(#1,\, #2\right)}

\newcommand{\dist}{\mathrm{dist}}

\newcommand{\N}{\mathbb{N}}

\newcommand{\R}{\mathbb{R}}

\newcommand{\II}{\mathcal{I}}

\newcommand{\TT}{\mathcal{T}}

\newcommand{\e}{\mathrm{e}}
\renewcommand{\epsilon}{\varepsilon}

\newcommand{\set}[1]{\left\{#1\right\}}
\newcommand{\tuple}[1]{\left(#1\right)} 

\newcommand{\tp}{\tuple}

\def\*#1{\boldsymbol{#1}} 
\def\+#1{\mathcal{#1}} 
\def\-#1{\mathrm{#1}} 
\def\=#1{\mathbb{#1}} 
\def\!#1{\mathfrak{#1}} 

\def\oPr{\mathop{\mathrm{Pr}}}
\renewcommand{\Pr}[2][]{
  \oPr_{#1}\left[#2\right] } 
  
\newcommand{\SI}{\mathsf{SI}}

\title{Improved Mixing of Critical Hardcore Model}
\author{Zongchen Chen \thanks{School of Computer Science, College of Computing, Georgia Institute of Technology, Atlanta, Georgia, USA. Email: \texttt{chenzongchen@gatech.edu}.} 
\and Tianhui Jiang \thanks{Zhiyuan College, Shanghai Jiao Tong University, Shanghai, China. Email: \texttt{shrimp2004@sjtu.edu.cn}.}}
\date{January 7, 2026}

\pgfplotsset{compat=1.18} 
\begin{document}
	
\maketitle
\thispagestyle{empty}
\begin{abstract}
	The hardcore model is one of the most classic and widely studied examples of undirected graphical models. Given a graph $G$, the hardcore model describes a Gibbs distribution of $\lambda$-weighted independent sets of $G$. In the last two decades, a beautiful computational phase transition has been established at a precise threshold $\lambda_c(\Delta)$ where $\Delta$ denotes the maximum degree, where the task of sampling independent sets transitions from polynomial-time solvable to computationally intractable. 
	We study the critical hardcore model where $\lambda = \lambda_c(\Delta)$
	and show that the Glauber dynamics, a simple yet popular Markov chain algorithm, mixes in $\tilde{O}(n^{4+O(1/\Delta)})$ time on any $n$-vertex graph of maximum degree $\Delta\geq3$, significantly improving the previous upper bound $\tilde{O}(n^{12.88+O(1/\Delta)})$ by the recent work \cite{CCYZ24}.
	Our improvement comes from an optimal bound on the $\ell_\infty$-spectral independence for the hardcore model at all subcritical fugacity $\lambda < \lambda_c(\Delta)$.
\end{abstract}
	
\newpage
	
\setcounter{page}{1}
	
\section{Introduction}

The \emph{hardcore model} is one of the most fundamental undirected graphical models that has been extensively studied in statistical physics, social science, probability theory, combinatorics, and computer science.

Given a graph $G=(V,E)$, we let $\II(G)$ denote the collection of all independent sets of $G$, where we recall that an independent set is a subset of vertices inducing no edges. 
The \emph{Gibbs distribution} $\mu_{G,\lambda}$ associated with the hardcore model on $G$ is parameterized by a vertex weight $\lambda > 0$ called the \emph{fugacity}. 
Each independent set $\sigma \in \II(G)$ receives a probability density given by
\begin{align*}
	\mu_{G,\lambda}(\sigma) = \frac{\lambda^{|\sigma|}}{Z_{G,\lambda}},
\end{align*}
where $Z_{G,\lambda}$ is a normalizing constant call the \emph{partition function} and is defined as
\begin{align*}
	Z_{G,\lambda} = \sum_{\sigma \in \II(G)} \lambda^{|\sigma|}.
\end{align*}

Perhaps the most amazing property of the hardcore model is the phase transition phenomenon associated with it. In fact, the hardcore model was originally proposed by statistical physicists to study and understand the phase transition in systems of hardcore gas particles. 
Let $\Delta \ge 3$ denote the maximum degree of the underlying graph.
The tree-uniqueness threshold $\lambda_c(\Delta) := \frac{(\Delta-1)^{\Delta-1}}{(\Delta-2)^\Delta}$ characterizes the uniqueness of the hardcore Gibbs measure on the infinite $\Delta$-regular tree. Furthermore, it also describes the existence of long-range correlations. 
Let each vertex be associated with a Bernoulli random variable, called the \emph{spin}, indicating whether the vertex is \emph{occupied} (i.e., included in the independent set) or \emph{unoccupied} (i.e., not included in the independent set).
Then, for small fugacity $\lambda \le \lambda_c(\Delta)$ the configuration at distance $\ell$ from the root has a vanishing influence on the root as $\ell$ tends to infinity, while for large fugacity $\lambda > \lambda_c(\Delta)$ the correlation is always bounded away from zero.

In the past two decades, a beautiful computational phase transition has been fully established for the problem of sampling from the hardcore model on graphs of maximum degree $\Delta$, precisely around the uniqueness threshold $\lambda_c(\Delta)$.
For $\lambda < \lambda_c(\Delta)$, there exist deterministic approximate counting algorithms for estimating the partition function \cite{weitz2006counting,Bar06,peters2019conjecture}, which in turn gives approximate samplers via standard reduction. 
Meanwhile, for $\lambda > \lambda_c(\Delta)$, no polynomial-time approximate counting and sampling algorithms exist assuming $\textsf{RP} \neq \textsf{NP}$ \cite{sly2010computational,sly2012computational,GSV16}.

While all deterministic approximate counting algorithms run in polynomial time, they suffer from a pretty slow runtime.
For example, Weitz's algorithm \cite{weitz2006counting} runs in time $n^{O( \frac{1}{\delta} \log \Delta)}$ where $\Delta$ denotes the maximum degree and $\delta \in (0,1)$ the slackness of the fugacity (i.e., $\lambda = (1-\delta) \lambda_c(\Delta)$).
In practice, Markov chain Monte Carlo (MCMC) algorithms provide a simpler and significantly faster method for generating random samples from high-dimensional distributions, including the hardcore model studied in this work. 
Among them, the \emph{Glauber dynamics} (also known as the \emph{Gibbs sampler}) is one of the most important and popular examples.
The Glauber dynamics performs a random walk in the space $\II(G)$ of independent sets and, in each step, either stays the same or moves to an adjacent set whose Hamming distance to the current set is 1.
More specifically, from the current independent set $\sigma_t \in \II(G)$, the algorithm picks a vertex $v \in V$ uniformly at random and updates its spin: Let $S = \sigma_t \setminus \{v\}$; if $S \cup \{v\} \notin \II(G)$ then set $\sigma_{t+1} = S = \sigma_t$; otherwise, set $\sigma_{t+1} = S \cup \{v\}$ with probability $\lambda/(1+\lambda)$ and, mutually exclusively, set $\sigma_{t+1} = S$ with probability $1/(1+\lambda)$.

Let $P_{\mathrm{GD}}$ denote the transition matrix of the Glauber dynamics.
From basic Markov chain theories it is easy to show that the Glauber dynamics $P_{\mathrm{GD}}$ is irreducible, aperiodic, and reversible with respect to the Gibbs distribution $\mu_{G,\lambda}$, which is the unique stationary distribution (i.e., $\mu_{G,\lambda} P_{\mathrm{GD}} = \mu_{G,\lambda}$).
The mixing time of Glauber dynamics is defined as
\begin{align*}
	T_{\mathrm{mix}}(P_{\mathrm{GD}}) = \max_{\sigma_0 \in \II(G)} \min_{t \in \N} \left\{ \TV{P^t(\sigma_0, \cdot)}{\mu_{G,\lambda}} \le \frac{1}{4} \right\},
\end{align*}
where $\sigma_0$ is the initial independent set, $P_{\mathrm{GD}}^t(\sigma_0, \cdot)$ is the distribution of the chain after $t$ steps when starting from $\sigma_0$, and $d_{\mathrm{TV}}(\cdot,\cdot)$ denotes the total variation distance.

In the past years, exciting progress has been made in understanding the mixing time of Glauber dynamics for the hardcore model. 
Anari, Liu, and Oveis Gharan introduced a highly powerful technique known as \emph{spectral independence} \cite{anari2020spectral}, leading to significant advancements in this area, including resolutions to major open problems regarding mixing properties.
We refer to \cite{liu2023spectral,stefankovic2023lecture} for a thorough introduction of this technique.
In the subcritical regime (i.e., $\lambda < \lambda_c(\Delta)$), the mixing time of the Glauber dynamics was shown to be nearly linear $O(n \log n)$ \cite{anari2020spectral,chen2021optimal,CFYZ22optimal,chen2022localization}. 
Meanwhile, it was long known that in the supercritical regime (i.e., $\lambda > \lambda_c(\Delta)$), the mixing time could be exponentially large $\exp(\Omega(n))$ as witnessed by random $\Delta$-regular bipartite graphs \cite{mossel2009hardness}.

In a very recent work \cite{CCYZ24}, the mixing property is further investigated at the critical point (i.e., $\lambda = \lambda_c(\Delta)$).
For the upper bound, the mixing time of Glauber dynamics is $\tilde{O}(n^{2+4\e+O(1/\Delta)})$ on any $n$-vertex graph of maximum degree $\Delta$. 
For the lower bound, there exists an infinite sequence of graphs such that the mixing time is $\Omega(n^{4/3})$, which is, in particular, super-linear. 

In this work, we present an improved mixing time upper bound for the Glauber dynamics on the critical hardcore model.

\begin{theorem}
	\label{thm:main}
	For any $n$-vertex graph $G=(V,E)$ of maximum degree $\Delta \ge 3$, the Glauber dynamics for the hardcore model on $G$ with fugacity $\lambda = \lambda_c(\Delta)$ satisfies
	\begin{align*}
		T_{\mathrm{mix}}(P_{\mathrm{GD}}) =  O\left( n^{4+\frac{4}{\Delta-2}} \log \Delta\right).
	\end{align*}
\end{theorem}

Our mixing time upper bound $\tilde{O}(n^{4 + O(1/\Delta)})$ significantly improves over the $\tilde{O}(n^{12.88+O(1/\Delta)})$ mixing time previously established in \cite{CCYZ24}.

Similar to \cite{CCYZ24}, \cref{thm:main} is proved via the spectral independence framework.
Our main contribution is to establish an optimal bound on $\ell_\infty$-spectral independence (see \cref{def:SI}) in the whole subcritical regime $\lambda = (1-\delta)\lambda_c(\Delta)$ where $\delta \in (0,1)$.

\begin{remark}
    In the previous version of the paper, which appeared in the proceedings of \emph{Approximation, Randomization, and Combinatorial Optimization, Algorithms and Techniques (APPROX/RANDOM 2025)}, we claimed a mixing time upper bound of $\tilde{O}(n^{7.44+O(1/\Delta)})$ by establishing the $\ell_\infty$-spectral independence with constant $C_\SI^\infty = O(\sqrt{n})$ \emph{at criticality}. 
    However, we later discovered a fatal error in this claim.
    The claimed improvement in our previous version relied crucially on a lemma in the previous work \cite{CCYZ24} of the first author, which, informally speaking, states that one can reduce proving $\ell_\infty$-spectral independence on general graphs (more precisely, a stronger notion known as \emph{coupling independence} \cite{chen2023near}) to establishing a truncated version of it on trees. We later found a serious loophole in the proof of the lemma. In fact, we found that this lemma is incorrect and the claimed comparison of coupling independence between general graphs and trees does not hold, for example, in the ferromagnetic Ising model. 
    This invalidates the $\ell_\infty$-spectral independence of order $O(\sqrt{n})$ we previously claimed.

    In this revision, we adopt a completely different approach to obtain an improved mixing time upper bound $\tilde{O}(n^{4+O(1/\Delta)})$, which is even better than what we previously claimed. We focus on the $\ell_\infty$-spectral independence \emph{in the subcritical regime} rather than at criticality, and establish an \emph{optimal} constant (see \cref{thm:SI-subcritical,thm:SI-lb}), from which \cref{thm:main} readily follows.

    While we still believe $\ell_\infty$-spectral independence of order $O(\sqrt{n})$ should hold at criticality, we are unable to fix the lemma from \cite{CCYZ24} or find an alternative approach. If one could prove it, then the mixing time upper bound would be further brought down to $\tilde{O}(n^{3+O(1/\Delta)})$.
\end{remark}

\section{Preliminaries}
\label{sec:pre}

\subsection{Spectral independence}
\label{subsec:SI}
The core result of this work is to establish an optimal bound on the $\ell_\infty$-spectral independence for the subcritical hardcore model, from which \cref{thm:main} readily follows by sophisticated spectral independence techniques that have been developed in a recent line of works.

The following notion of influences is needed to define the meaning of spectral independence. 
\begin{definition}[Influence, \cite{anari2020spectral}]
	Let $\mu$ be a distribution over $\{0,1\}^n$, and let $A = \{i \in [n]: \Pr[\mu]{\sigma_i = 0} > 0 \land \Pr[\mu]{\sigma_i = 1} > 0\}$ be the set of unfixed coordinates.
	For any $i,j \in A$, define the \emph{(pairwise) influence} from $i$ to $j$ as
	\begin{align*}
		\Psi_\mu(i,j) := \Pr[\sigma\sim\mu]{\sigma_j = 1 \mid \sigma_i = 1} - \Pr[\sigma\sim\mu]{\sigma_j = 1 \mid \sigma_i = 0}.
	\end{align*}
	Note that $\Psi_\mu(i,i)=1$ for any $i \in A$.
	Further, let $\Psi_\mu$ be an $|A| \times |A|$ influence matrix with entries defined as above.
\end{definition}

The influence matrix $\Psi_\mu$ is not symmetric in general.
We remark that all eigenvalues of the influence matrix $\Psi_\mu$ are real; see, e.g., \cite{anari2020spectral}.

For a distribution $\mu$ over $\{0,1\}^n$, a \emph{pinning} $\tau \in \{0,1\}^\Lambda$ is a partial configuration on a subset of coordinates $\Lambda \subseteq [n]$ with positive density, i.e., $\Pr[\sigma\sim\mu]{\sigma_\Lambda = \tau}>0$.
Given a pinning $\tau$, we let $\mu^{\Lambda \gets \tau}$ denote the conditional distribution where the configuration on $\Lambda$ is fixed as $\tau$.
When the subset $\Lambda$ is clear or unimportant, we write $\mu^\tau = \mu^{\Lambda \gets \tau}$ for simplicity.

\begin{definition}[Spectral independence, \cite{anari2020spectral}]
\label{def:SI}
	Let $\mu$ be a distribution over $\{0,1\}^n$.
	We say $\mu$ satisfies \emph{spectral independence} with constant $C_\SI$ if for any pinning $\tau$, it holds
	\begin{align*}
		\lambda_{\-{max}}(\Psi_{\mu^\tau}) \le C_\SI,
	\end{align*}
    where $\lambda_{\-{max}}(\Psi_{\mu^\tau})$ denotes the maximum eigenvalue of $\Psi_{\mu^\tau}$.
	We say $\mu$ satisfies \emph{$\ell_\infty$-spectral independence} with constant $C^\infty_\SI$ if for any pinning $\tau$, it holds
	\begin{align*}
		\norm{\Psi_{\mu^\tau}}_\infty := 
		\max_{i \in A_{\mu^\tau}} \sum_{j \in A_{\mu^\tau}} \left| \Psi_{\mu^\tau}(i,j) \right|
		\le C^\infty_\SI.
	\end{align*}
\end{definition}

Since $\lambda_{\-{max}}(\Psi_{\mu^\tau}) \le \norm{\Psi_{\mu^\tau}}_\infty$, $\ell_\infty$-spectral independence with constant $C$ implies spectral independence with the same constant $C$.  

In the setting of the hardcore model, the influences describe the correlation between two vertices, represented as Bernoulli random variables indicating whether the vertices are occupied.
Roughly speaking, the influence of one vertex on the other represents the difference of the marginal distribution on the second vertex when flipping the first vertex from occupied to unoccupied.
Meanwhile, spectral independence describes a decay of correlation phenomenon in a novel spectral way.

\subsection{Mixing via spectral independence}
\label{subsec:SI-mixing}

The mixing time of Glauber dynamics can be derived from spectral independence.
In this subsection, we aim to provide a minimal introduction and apply the main spectral independence technique as a black box.
We refer the readers to \cite{liu2023spectral,CCYZ24,chen2022localization} for more background and details.

The following proposition shows that rapid mixing of the critical hardcore model can be derived from spectral independence in the whole subcritical regime.

\begin{proposition}[Spectral independence implies rapid mixing, \cite{CCYZ24,chen2022localization}]
	\label{prop:SS-mixing-full}
	Consider the hardcore model on an $n$-vertex graph $G=(V,E)$ of maximum degree $\Delta \ge 3$. 
	Suppose that for any $\lambda = (1-\delta)\lambda_c(\Delta)$ where $\delta \in (0,1)$, the Gibbs distribution $\mu_{G,\lambda}$ satisfies spectral independence with constant $\rho/\delta$, where $\rho > 0$ is an absolute constant.
	Then, for $\lambda = \lambda_c(\Delta)$, the mixing time of the Glauber dynamics for $\mu_{G,\lambda_c(\Delta)}$ satisfies
	\begin{align*}
		T_{\mathrm{mix}}(P_{\mathrm{GD}}) = O\left( n^{2+\rho} \log \Delta \right).
	\end{align*}
\end{proposition}

\cref{prop:SS-mixing-full} can be deduced from the \emph{localization scheme} framework introduced in \cite{chen2022localization}. Our version here was stated and proved in \cite{CCYZ24}.
We provide a proof sketch of \cref{prop:SS-mixing-full} in \cref{subsec:pf-SI-mixing} for completeness.

\subsection{Establishing spectral independence}

To establish $\ell_\infty$-spectral independence (which then implies spectral independence), it suffices to consider the sum of absolute influences on certain associated trees known as \emph{self-avoiding walk trees} \cite{weitz2006counting}. The formal definition and construction of these trees are omitted in this paper as we only need their existence, and we refer interested readers to the works \cite{weitz2006counting,chen2023rapid}.

\begin{proposition}[\cite{chen2023rapid}] \label{prop:SI-graph-tree}
	Consider the hardcore model on an $n$-vertex graph $G$ of maximum degree $\Delta \ge 3$ with fugacity $\lambda > 0$. 
	For any $u \in V$, there exists a tree $T = T_{\mathrm{SAW}}(G,u)$ rooted at $r$ with maximum degree at most $\Delta$, such that 
	\begin{align*}
		\sum_{v \in V(G)} \left| \Psi_{\mu_{G,\lambda}}(u,v) \right|
		\le \sum_{v \in V(T)} \left| \Psi_{\mu_{T,\lambda}}(r,v) \right|,
	\end{align*}
    where $V(G)$ and $V(T)$ denote the vertex set of $G$ and $T$, respectively.
\end{proposition}

Hence, to establish $\ell_\infty$-spectral independence with constant $C$ on all \emph{graphs} of maximum degree $\Delta$, it suffices to prove the absolute influence sum of the root is at most $C$ for all \emph{trees} of maximum degree $\Delta$. Namely, via \cref{prop:SI-graph-tree} we reduce our problem from general graphs to only trees.

\section{Optimal \texorpdfstring{$\ell_\infty$-Spectral Independence}{ℓ∞-Spectral Independence}}

In this section, we establish an optimal $\ell_\infty$-spectral independence constant for the hardcore model in the tree-uniqueness regime.

\begin{theorem}[Upper bound]
	\label{thm:SI-subcritical}
	Suppose $\Delta = d+1 \ge 3$ is an integer and $\delta\in(0,1)$ is a real number.
	Consider the hardcore model on a graph $G=(V,E)$ of maximum degree $\Delta$ with fugacity $\lambda = (1-\delta) \lambda_c(\Delta)$. 
	Then, the Gibbs distribution $\mu_{G,\lambda}$ satisfies $\ell_\infty$-spectral independence with constant 
	\begin{align*}
		C^\infty_\SI = \frac{1+\hat{x}}{1-d\hat{x}} 
		\le \frac{2}{\delta} \left(1+\frac{2}{d-1}\right),
	\end{align*}
	where $\hat x = \hat{x}(d,\lambda) \in (0,1/d)$ is the unique fixed point of the function $F_{d,\lambda}(x) = \frac{\lambda(1-x)^d}{1+\lambda(1-x)^d}$ for $x\in[0,1]$.
\end{theorem}

\cref{thm:main} follows directly from \cref{thm:SI-subcritical} and \cref{prop:SS-mixing-full}.

\begin{proof}[Proof of \cref{thm:main}]
    For any graph $G$ of maximum degree $\Delta\geq3$, by \cref{thm:SI-subcritical}, for any $\delta\in(0,1)$, $\mu_{G,(1-\delta)\lambda_c(\Delta)}$ satisfies $\ell_\infty$-spectral independence with constant $\frac{\rho}{\delta}$, where $\rho = 2 \left(1+\frac{2}{\Delta-2}\right)$.
    Then, by \cref{prop:SS-mixing-full}, the mixing time of the Glauber dynamics for $\mu_{G,\lambda_c(\Delta)}$ satisfies
	\begin{align*}
		T_{\mathrm{mix}}(P_{\mathrm{GD}}) 
        =  O\left( n^{2+\rho} \log \Delta\right)
        =  O\left( n^{4+\frac{4}{\Delta-2}} \log \Delta\right),
	\end{align*}
    as desired.
\end{proof}

Before presenting the proof of \cref{thm:SI-subcritical}, we first explain the optimality of our constant $C^\infty_\SI$ in $\ell_\infty$-spectral independence.
Let $\mathbb{T}_d$ denote the infinite $d$-ary tree and $\bar{\mathbb{T}}_\Delta$ be the infinite $\Delta$-regular tree, where $\Delta = d+1\ge 3$. Since $\lambda<\lambda_c(\Delta)$, there is a unique infinite-volume Gibbs measure for the hardcore model on either $\mathbb{T}_d$ or $\bar{\mathbb{T}}_\Delta$ with fugacity $\lambda$ \cite{spitzer1975markov,kelly1985stochastic,weitz2006counting}.
For the Gibbs measure on $\mathbb{T}_d$, the occupancy probability of the root is equal to $\hat{x}$, the unique fixed point of the associated tree recurrence $F_{d,\lambda}$; see \cref{fact:tree-recursion}.
Furthermore, it can be shown that for the Gibbs measure on $\bar{\mathbb{T}}_\Delta$, the absolute influence of the root $r$ on a vertex $v$ is expressed by this occupancy probability $\hat{x}$, specifically,
\begin{align*}
	\left|\Psi_{\mu_{\bar{\mathbb{T}}_\Delta,\lambda}}(r,v)\right| = \hat{x}^{\dist(r,v)},
\end{align*}
where $\dist(r,v)$ is the distance between $r$ and $v$; see \cref{fact:prod}.
Therefore, the sum of absolute influences of the root can be calculated as
\begin{align*}
	\sum_{v \in V(\bar{\mathbb{T}}_\Delta)} \left|\Psi_{\mu_{\bar{\mathbb{T}}_\Delta,\lambda}}(r,v)\right|
	= 1 + \sum_{k=1}^\infty (d+1) d^{k-1} \hat{x}^k
	= \frac{1+\hat{x}}{1-d\hat{x}}.
\end{align*}
Hence, the $\ell_\infty$-spectral independence constant $C^\infty_\SI = \frac{1+\hat{x}}{1-d\hat{x}}$ from \cref{thm:SI-subcritical} is achieved by the infinite $\Delta$-regular tree $\bar{\mathbb{T}}_\Delta$.
In other words, \cref{thm:SI-subcritical} shows that, informally speaking, among all graphs of maximum degree $\Delta$, the infinite $\Delta$-regular tree $\bar{\mathbb{T}}_\Delta$ attains the largest constant for $\ell_\infty$-spectral independence.

As we focused on finite graphs, we provide the following rigorous result for lower bounds on $\ell_\infty$-spectral independence on finite trees, which essentially uses the sequence of Gibbs distributions on truncated regular trees to approximate the Gibbs measure on the infinite regular tree.
To be specific, for each $h\geq 1$, we consider the rooted $\Delta$-regular tree truncated at depth $h$, denoted by $T_{\Delta,h}$, which is obtained from the infinite $\Delta$-regular tree by removing all vertices at distance greater than $h$ from the root.
Equivalently, $T_{\Delta,h}$ is the tree with leaves at depth $h$, the root having $\Delta$ children, and every non-root, non-leaf vertex having $\Delta-1$ children.

\begin{theorem}[Lower bound]
	\label{thm:SI-lb}
	Suppose $\Delta = d+1 \ge 3$ is an integer and $\delta\in(0,1)$ is a real number.
	For $h \ge 1$, consider the hardcore model on $T_{\Delta,h}$, the rooted $\Delta$-regular tree truncated at depth $h$, with fugacity $\lambda = (1-\delta) \lambda_c(\Delta)$. 
	If $C_{\Delta,h,\lambda}$ is the optimal $\ell_\infty$-spectral independence constant for the Gibbs distribution $\mu_{T_{\Delta,h},\lambda}$, then we have
	\begin{align*}
		\lim_{h\to \infty} C_{\Delta,h,\lambda} = \frac{1+\hat{x}}{1-d\hat{x}},
	\end{align*}
	where $\hat{x}$ is the unique fixed point of $F_{d,\lambda}$ defined in \cref{thm:SI-subcritical}.
\end{theorem}

We suspect that the constant $\frac{1+\hat{x}}{1-d\hat{x}}$ is actually optimal even for the standard spectral independence. The sequence of graphs that would achieve this constant for spectral independence is random regular (symmetric) bipartite graphs, where it is known that the local distribution within a ball of constant radius at a random vertex converges to that on the infinite regular tree \cite{sly2012computational}.

\subsection{Upper bound}

In this subsection, we establish \cref{thm:SI-subcritical}.
With \cref{prop:SI-graph-tree} at hand, we need to establish the following bound on the absolute influence sum on trees, from which \cref{thm:SI-subcritical} readily follows.
\begin{lemma}\label{lem:sub-SI-tree}
    Let $T=(V,E)$ be a tree rooted at $r$ with maximum degree at most $\Delta = d+1$, where $d \ge 2$. 
    For $\delta\in (0,1)$, consider the hardcore model on $T$ with fugacity $\lambda = (1-\delta) \lambda_c(\Delta)$. 
    Then, it holds that
    \begin{align*}
        \sum_{v \in V} \left| \Psi_{\mu_{T,\lambda}}(r,v) \right|
        \le \frac{1+\hat{x}}{1-d\hat{x}},
    \end{align*}
    where $\hat{x}$ is the unique fixed point of $F_{d,\lambda}$ defined in \cref{thm:SI-subcritical}.
\end{lemma}

Let $T=(V,E)$ be a tree rooted at $r$. 
For every vertex $v\in V$, let $T_v$ denote the subtree of $T$ rooted at $v$ that consists of all descendants of $v$; in particular, $T_r = T$.
For any $v\in V$, let $L(v)$ denote the set of children of $v$ in $T$.
Let $L_k(r)$ denote the set of vertices that are at distance $k$ from the root $r$.

Consider the hardcore model on $T$ with fugacity $\lambda>0$.
For each vertex $v \in V$, let $p_v$ denote the probability that $v$ is occupied in the hardcore model on the subtree $T_v$ rooted at $v$, i.e., 
\begin{align*}
	p_v := \Pr[\mu_{T_v,\lambda}]{\sigma_v = 1}.
\end{align*}
Furthermore, we denote the sum of absolute influences of the root as
\begin{align*}
	\Phi(T,\lambda) := \sum_{v \in V} \left| \Psi_{\mu_{T,\lambda}}(r,v) \right|.
\end{align*}

We need the following standard facts regarding occupancy probabilities and influences on trees.
\begin{fact}[Tree recursion, \cite{weitz2006counting}]
	\label{fact:tree-recursion}
	Consider the hardcore model on a tree $T=(V,E)$ rooted at $r$ with fugacity $\lambda>0$.
	For any $v\in V$, we have
	\begin{align*}
		\frac{p_v}{1-p_v} = \lambda \prod_{w\in L(v)} (1-p_w).
	\end{align*}
\end{fact}

\begin{fact}[\cite{chen2023rapid}]
	\label{fact:prod}
	Consider the hardcore model on a tree $T=(V,E)$ rooted at $r$ with fugacity $\lambda>0$.
	For any $v\in V$, if $r=u_0,u_1,\cdots,u_m=v$ is the unique path from $r$ to $v$ in $T$, then we have
	\begin{align*}
		\left| \Psi_{\mu_{T,\lambda}}(r,v) \right| = \prod_{i=1}^m p_{u_i}.
	\end{align*}
\end{fact}

We first consider a slightly different version of \cref{lem:sub-SI-tree}, where every vertex in the tree has at most $d$ children; that is, every vertex except for the root has degree at most $\Delta=d+1$, and the root has degree at most $\Delta-1=d$.
This makes it easier for us to adopt a recursive argument.

\begin{lemma}\label{lem:sub-tree-SI}
	Let $T=(V,E)$ be a tree rooted at $r$ such that each vertex has at most $d$ children, where $d \ge 2$. 
	For $\delta\in (0,1)$, consider the hardcore model on $T$ with fugacity $\lambda = (1-\delta) \lambda_c(\Delta)$ where $\Delta = d+1$. 
	Then, it holds that
	\begin{align*}
		\Phi(T,\lambda)
		\le \frac{1}{1-d\hat x},
	\end{align*}
	where $\hat{x}$ is the unique fixed point of $F_{d,\lambda}$ defined in \cref{thm:SI-subcritical}.
\end{lemma}

\begin{proof}
	Let $\TT_d$ be the family of all rooted trees such that each vertex has at most $d$ children.
	For $\lambda = (1-\delta) \lambda_c(\Delta)$ where $\delta\in (0,1)$,
	we define 
	\begin{align*}
		\Phi^* = \Phi^*(d,\lambda) := \sup_{T\in\TT_d} \Phi(T,\lambda).
	\end{align*}
    We note that $\Phi^* < \infty$, since $\Phi(T,\lambda)\leq \frac{32}{\delta}$ for all $T\in\TT_d$ (see \cite{chen2023rapid}).
    We need to show that $\Phi^* \le \frac{1}{1-d\hat x}$.

    Recall that $T_w$ denotes the subtree of $T$ rooted at $w$ that consists of all descendants of $w$, and that $L_k(r)$ denotes the set of vertices that are at distance $k$ from the root $r$.
    By \cref{fact:prod}, we have
    \begin{align}\label{eq:W(T)}
        \Phi(T,\lambda) &= \left| \Psi_{\mu_{T,\lambda}}(r,r) \right| 
        + \sum_{v\in L_1(r)} \left| \Psi_{\mu_{T,\lambda}}(r,v) \right| 
        + \sum_{k=2}^{\infty}\sum_{v\in L_k(r)} \left| \Psi_{\mu_{T,\lambda}}(r,v) \right| \nonumber\\
        &= 1+\sum_{v \in L(r)} p_v+\sum_{v\in L(r)} \sum_{w\in L(v)} p_v p_w \Phi(T_w,\lambda)\nonumber\\
        &\leq 1+\sum_{v \in L(r)} \tp{ p_v+p_v\sum_{w\in L(v)}  p_w \Phi^*}, 
    \end{align}
    where the last inequality follows from $T_w\in \TT_d$. 
    By tree recursion \cref{fact:tree-recursion}, we have
\begin{align*}
    \frac{p_v}{1-p_v} = \lambda \prod_{w\in L(v)} (1-p_w)\leq \lambda \tp{1-\frac{1}{d}\sum_{w\in L(v)}p_w}^d = \lambda (1-\bar p_v)^d,
\end{align*}
where $\bar p_v = \frac{1}{d}\sum_{w\in L(v)}p_w$, and the inequality follows from the AM-GM inequality and $|L(v)|\leq d$.
    Therefore, we obtain
\begin{align*}
    p_v\leq\frac{\lambda (1-\bar p_v)^d}{1+\lambda (1-\bar p_v)^d},
\end{align*}
and consequently,
\begin{align*}
    p_v\sum_{w\in L(v)}p_w
    \leq\frac{\lambda (1-\bar p_v)^d}{1+\lambda (1-\bar p_v)^d}\sum_{w\in L(v)}p_w
    =d\bar p_v\frac{\lambda (1-\bar p_v)^d}{1+\lambda (1-\bar p_v)^d}.
\end{align*}
Combining with \cref{eq:W(T)}, we get
\begin{align*}
	\Phi(T,\lambda) &\leq 1+\sum_{v \in L(r)} \tp{ \frac{\lambda (1-\bar p_v)^d}{1+\lambda (1-\bar p_v)^d}+d\bar p_v\frac{\lambda (1-\bar p_v)^d}{1+\lambda (1-\bar p_v)^d} \Phi^*},
\end{align*} 
Suppose $q \in [0,1]$ maximizes the function $a(x):= \frac{\lambda (1-x)^d}{1+\lambda (1-x)^d}+dx\frac{\lambda (1-x)^d}{1+\lambda (1-x)^d} \Phi^* $.
Taking supremum over all $T\in\TT_d$, we have
\begin{align}\label{eq:Phi*-bound}
    \Phi^* = \sup_{T\in \TT_d} \Phi(T,\lambda) 
    \le 1 + d \tp{ \frac{\lambda (1-q)^d}{1+\lambda (1-q)^d} + dq\frac{\lambda (1-q)^d}{1+\lambda (1-q)^d} \Phi^*}.
\end{align}
We claim that
\begin{align}\label{eq:valid}
	d^2 q \frac{\lambda (1-q)^d}{1+\lambda (1-q)^d} < 1 
	\quad\Longleftrightarrow\quad
	\left(d^2 q - 1\right) \lambda (1-q)^d < 1. 
\end{align}
If $d^2 q < 1$, then \cref{eq:valid} is trivial.
Otherwise, \cref{eq:valid} follows from an application of the AM--GM inequality: 
\begin{align*}
	\left(d^2 q - 1\right) \lambda (1-q)^d 
	&= \lambda d \left(d q - \frac{1}{d}\right) (1-q)^d \\
	&\le \lambda d \left( \frac{(d q - \frac{1}{d}) + d(1-q)}{d+1} \right)^{d+1} 
	= \lambda \frac{(d-1)^{d+1}}{d^d} = \frac{\lambda}{\lambda_c(\Delta)} < 1.
\end{align*}
Combining \cref{eq:Phi*-bound,eq:valid}, we then deduce that
\begin{align*}
    \Phi^* &\leq \frac{1+ d \frac{\lambda (1-q)^d}{1+\lambda (1-q)^d}}{1-d^2q\frac{\lambda (1-q)^d}{1+\lambda (1-q)^d}}
    = \frac{1+(d+1)\lambda(1-q)^d}{1-(d^2q-1)\lambda(1-q)^d}
    = f(q),
\end{align*}
where the function $f:[0,1] \to \R$ is defined as
\begin{align*}
	f(x) := \frac{1+(d+1)\lambda(1-x)^d}{1-(d^2x-1)\lambda(1-x)^d}.
\end{align*}
In particular, it holds that
\begin{align*}
	\Phi^* \leq \max_{x\in[0,1]} f(x).
\end{align*}
In the remaining proof, we aim to show that
\begin{align*}
	\max_{x\in[0,1]} f(x) = f(\hat x) = \frac{1}{1-d\hat x},
\end{align*}
thus completing the proof of the lemma.

Observe that
\begin{align*}
    \frac{1}{f(x)}
    =\frac{1-(d^2x-1)\lambda(1-x)^d}{1+(d+1)\lambda(1-x)^d}
    =1-\frac{d\lambda(1+dx)}{(1-x)^{-d}+(d+1)\lambda}.
\end{align*}
Define $g:[0,1] \to \R\cup\{+\infty\}$ as  
\begin{align*}
    g(x)=\frac{(1-x)^{-d}+(d+1)\lambda}{1+dx}.
\end{align*}
Then, it holds
\begin{align*}
    \frac{1}{f(x)}=1-\frac{d\lambda}{g(x)}.
\end{align*}
We calculate that
\begin{align*}
    g'(x)&=\frac{d(1-x)^{-(d+1)}(1+dx)-d((1-x)^{-d}+(d+1)\lambda)}{(1+dx)^2}\\
    &=\frac{d}{(1+dx)^2}\tp{\frac{1+dx-(1-x)}{(1-x)^{d+1}}-(d+1)\lambda}\\
    &=\frac{d(d+1)}{(1+dx)^2}\tp{\frac{x}{(1-x)^{d+1}}-\lambda}.
\end{align*}
Since the first factor of $g'(x)$ is always positive, the sign of $g'(x)$ depends on the second factor, i.e., $\frac{x}{(1-x)^{d+1}}-\lambda$.
Let $h(x)=\frac{x}{(1-x)^{d+1}}-\lambda$.
Clearly, $h(x)$ is increasing on $[0,1]$, $h(0)<0$, $\lim_{x\uparrow1} h(x)=+\infty$. 
By the Intermediate Value Theorem, there exists a unique zero of $h(x)$ on $(0,1)$, denoted by $\hat x$.
Observe that $h(\hat x)=0$ is equivalent to $\hat x = \frac{\lambda (1-\hat x)^d}{1+\lambda (1-\hat x)^d}$, i.e., $\hat x$ is the unique fixed point of $F_{d,\lambda}(x)$. 
Therefore, we deduce that 
\begin{align*}
    \min_{x\in[0,1]} g(x) = g(\hat x)
    \quad\text{and}\quad
    \max_{x\in[0,1]} f(x) = f(\hat x).
\end{align*}
Finally, we compute that
\begin{align*}
    f(\hat x) = \frac{1+ d \frac{\lambda (1-\hat x)^d}{1+\lambda (1-\hat x)^d}}{1-d^2\hat x\frac{\lambda (1-\hat x)^d}{1+\lambda (1-\hat x)^d}}= \frac{1+d\hat x}{1-d^2\hat x^2} = \frac{1}{1-d\hat x},
\end{align*}
as claimed.
\end{proof}

We now present the proof of \cref{lem:sub-SI-tree}. 

\begin{proof}[Proof of \cref{lem:sub-SI-tree}]
	If the root $r$ has degree at most $d=\Delta-1$, then \cref{lem:sub-SI-tree} follows immediately from \cref{lem:sub-tree-SI}.
	In the following, we assume that $r$ has degree $\Delta=d+1$.
	For each $v \in L(r)$, let $T^v := T \setminus T_v$ denote the subtree obtained by removing the subtree $T_v$ rooted at $v$ from $T$.
	By the Markov property of the hardcore model on trees, namely, when the root $r$ is fixed, the conditional distributions on each subtree $T_v$ where $v \in L(r)$ are jointly independent,
	we observe that
	\begin{align*}
		\left| \Psi_{\mu_{T^v,\lambda}}(r,w) \right| = \left| \Psi_{\mu_{T,\lambda}}(r,w) \right|
	\end{align*}
	for all $v \in L(r)$ and $w \in V(T^v) = V \setminus V(T_v)$.
	Therefore, we deduce that
	\begin{align*}
		\sum_{v \in L(r)} \Phi(T^v,\lambda) = d \Phi(T,\lambda) + \left| \Psi_{\mu_{T,\lambda}}(r,r) \right|
		= d \Phi(T,\lambda) + 1.
	\end{align*}
	Notice that $T^v \in \TT_d$ for each $v$, and thereby we deduce from \cref{lem:sub-tree-SI} that
	\begin{align*}
		\Phi(T,\lambda) = \frac{1}{d} \left( \sum_{v \in L(r)} \Phi(T^v,\lambda) - 1 \right) \le \frac{1}{d} \left( \frac{d+1}{1-d\hat{x}} - 1 \right) = \frac{1+\hat{x}}{1-d\hat{x}},
	\end{align*}
	as claimed.
\end{proof}

Our proof of \cref{lem:sub-tree-SI} implicitly established the known fact that $\hat{x} < 1/d$ when $\lambda < \lambda_c(\Delta)$ \cite{li2013correlation}; in fact, $\hat{x} = 1/d$ when $\lambda = \lambda_c(\Delta)$.
This can be seen from \cref{eq:valid} by plugging in $q = \hat{x}$ and noticing $\frac{\hat{x}}{1-\hat{x}} = \lambda (1-\hat{x})^d$. 
In the following lemma, we give a more precise upper bound on $\hat{x}$, which allows us to bound the spectral independence constant by a function of the maximum degree $\Delta$ and the slackness parameter $\delta$ in \cref{thm:SI-subcritical}.

\begin{lemma} \label{lem:sub-fixed}
    Let $d \ge 2$ be an integer and $\delta \in (0,1)$ be a real number. Suppose $\lambda = (1-\delta) \lambda_c(\Delta)$ is the fugacity where $\Delta = d+1$.
    Let $\hat x$ be the unique fixed point of the function $F_{d,\lambda}(x) = \frac{\lambda(1-x)^d}{1+\lambda(1-x)^d}$ for $x\in[0,1]$.
    Then, it holds that
    \begin{align*}
        \hat x\leq \frac{1}{d}\tp{1-\frac{d-1}{2d}\delta}.
    \end{align*}
\end{lemma}

\begin{proof}
    Plugging $\lambda = (1-\delta) \lambda_c(\Delta)$ into $F_{d,\lambda}(\hat x)=\hat x$, we have
    \begin{align} \label{eq:fixed}
        \frac{(1-\delta)\lambda_c(\Delta)(1-\hat x)^d}{1+(1-\delta)\lambda_c(\Delta)(1-\hat x)^d}=\hat x.
    \end{align} 
    Fixing $\Delta$, there is a one-to-one correspondence between $\hat x$ and $\delta$ by \cref{eq:fixed}.
    We can view $\hat x = \hat x (\delta)$ as a function of $\delta$ and vice versa.
    Specifically, from \cref{eq:fixed}, we get
    \begin{align*}
        \delta(\hat x) = 1-\frac{1}{\lambda_c(\Delta)}\frac{\hat x}{(1-\hat x)^{d+1}}.
    \end{align*}
    Observe that $\delta(\hat x)$ is monotonically decreasing, with $\delta(0)=1$ and $\delta(\frac{1}{d})=0$.
    Differentiating both sides with respect to $\hat x$, we have
    \begin{align*}
        \delta'(\hat x) &= -\frac{1}{\lambda_c(\Delta)}\tp{\frac{1}{(1-\hat x)^{d+1}}+\frac{(d+1)\hat x}{(1-\hat x)^{d+2}}}\\
        &=-\frac{1}{\lambda_c(\Delta)}\frac{1+d\hat x}{(1-\hat x)^{d+2}} < 0.
    \end{align*}
    Notice that $\hat x(0)=\frac{1}{d}$ since $\delta(\frac{1}{d}) = 0$.
    By substituting, we have $\delta'(\frac{1}{d})=-\frac{2d^2}{d-1}$, from which we deduce that
    \begin{align*}
        \hat x'(0)=\frac{1}{\delta'(\frac{1}{d})}=-\frac{d-1}{2d^2}.
    \end{align*}
    Further, by the monotonicity of the function $\delta'(\hat{x})$, we know that $\delta''(\hat x)\leq0$ for all $\hat x\in[0,1/d]$.
    Therefore, for all $\delta\in(0,1)$, it holds
    \begin{align*}
        \hat x''(\delta)=-\frac{\delta''(\hat x)}{(\delta'(\hat x))^3} \leq 0.
    \end{align*}
    This implies that $\hat x(\delta)$ is concave on $(0,1)$, 
    and it follows that for all $\delta \in (0,1)$, we have
    \begin{align*}
        \hat x(\delta)\leq\hat x(0)+\delta \hat x'(0)=\frac{1}{d}-\frac{d-1}{2d^2}\delta=\frac{1}{d}\tp{1-\frac{d-1}{2d}\delta},
    \end{align*}
    which shows the inequality we desired.
\end{proof}

We end this subsection with the proof of \cref{thm:SI-subcritical}.
We introduce some notations needed for the proof.
Let $G=(V,E)$ be a graph.
For any $S \subseteq V$, let $\partial S$ denote the set of neighbors of $S$ in $G$, i.e., $\partial S = \set{ v \in V \setminus S \mid \exists u \in S, \set{u, v} \in E }$; and let $G[S]$ denote the subgraph induced in $G$ by $S$, i.e., the graph with vertex set $S$ and edge set consisting of all edges of $G$ that have both endpoints in $S$.

\begin{proof}[Proof of \cref{thm:SI-subcritical}]
    For any graph $G=(V,E)$ with maximum degree $\Delta$ and $\lambda=(1-\delta) \lambda_c(\Delta)$, we first show that $\mu_{G,\lambda}$ satisfies $\ell_\infty$-spectral independence with constant $\frac{1+\hat{x}}{1-d\hat{x}}$.
	Consider any pinning $\tau \in \{0,1\}^\Lambda$ satisfying $\Pr[\sigma\sim\mu_{G,\lambda}]{\sigma_\Lambda = \tau}>0$. For $i\in\set{0,1}$, let $S_i$ be the set of vertices that are pinned to $i$, i.e., $S_i=\set{v\in \Lambda\vert \tau_v=i}$.
    Then, the conditional Gibbs distribution $\mu_{G,\lambda}^{\tau}$ corresponds to the hardcore model on the induced subgraph $G'=G[V\setminus (S_0\cup S_1\cup\partial S_1)]$.
    It suffices to show
    \begin{align}\label{eq:induced-graph-SI}
\norm{\Psi_{\mu_{G',\lambda}}}_\infty\leq \frac{1+\hat{x}}{1-d\hat{x}}.
    \end{align}
    Since $G'$ has maximum degree at most $\Delta$, according to \cref{prop:SI-graph-tree}, for any $u\in V(G')$, there exists a tree $T = T_{\mathrm{SAW}}(G',u)$ rooted at $r$ with maximum degree at most $\Delta$, such that 
	\begin{align*}
		\sum_{v \in V(G')} \left| \Psi_{\mu_{G',\lambda}}(u,v) \right|
		\le \sum_{v \in V(T)} \left| \Psi_{\mu_{T,\lambda}}(r,v) \right| \leq \frac{1+\hat{x}}{1-d\hat{x}},
	\end{align*}
    where the last inequality follows from \cref{lem:sub-SI-tree}.
    Therefore, \cref{eq:induced-graph-SI} holds, and by definition, $\mu_{G,\lambda}$ satisfies $\ell_\infty$-spectral independence with constant $\frac{1+\hat{x}}{1-d\hat{x}}$.
    
    By \cref{lem:sub-fixed},
    \begin{align*}
        \frac{1+\hat{x}}{1-d\hat{x}}\leq \frac{1+\frac{1}{d}\tp{1-\frac{d-1}{2d}\delta}}{1-d\cdot\frac{1}{d}\tp{1-\frac{d-1}{2d}\delta}}
        =\frac{2}{\delta} \left(1+\frac{2}{d-1}\right)-1
        \leq \frac{2}{\delta} \left(1+\frac{2}{d-1}\right),
    \end{align*}
    which shows the inequality we desired.
\end{proof}

\subsection{Lower bound}

In this subsection, we prove \cref{thm:SI-lb}.
We first define the $t$-fold iteration of a mapping.
\begin{definition}[$t$-fold iteration]\label{def:t-fold-iter}
    Let $f:X \rightarrow X$ be a mapping.
    Define $f^{(t)}:X\rightarrow X$ for $t\in \N$ inductively by $f^{(0)} = \mathrm{id}$ and $f^{(t)} = f^{(t-1)} \circ f$ for $t \ge 1$.
    We call $f^{(t)}$ the \emph{$t$-fold iteration} of $f$.
\end{definition}
We are now ready to prove \cref{thm:SI-lb}.
\begin{proof}[Proof of \cref{thm:SI-lb}]
    Notice that $C_{\Delta,h,\lambda}\geq \Phi(T_{\Delta,h},\lambda)$.
    We first calculate $\Phi(T_{\Delta,h},\lambda)$.
    Let $V$ denote the vertex set of $T_{\Delta,h}$, and let $r$ denote its root.
    For every $v\in V$, we define 
    \begin{align*}
	p_v^{(h)} := \Pr[\mu_{T_v,\lambda}]{\sigma_v = 1},
    \end{align*}
where $T_v$ is the subtree of $T_{\Delta,h}$ rooted at $v$ that consists of all descendants of $v$.
Notice that for every $v\in L_k(r)$ with $1\leq k\leq h$, $T_v$ is the full $d$-ary tree with height $h-k$, where we recall that $L_k(r)$ denotes the set of vertices that are at distance $k$ from the root $r$.
Then by \cref{fact:tree-recursion}, we have
\begin{align}\label{eq:p_k^h}
    p_{v}^{(h)}=F^{(h-k+1)}_{d,\lambda}(0),
\end{align}
where $F_{d,\lambda}^{(h-k+1)}$ is the $(h-k+1)$-fold iteration of $F_{d,\lambda}$ (see \cref{def:t-fold-iter}).
For every $1\leq k\leq h$, we define 
\begin{align*}
    a_k^{(h)}:=\left| \Psi_{\mu_{T_{\Delta,h},\lambda}}(r,v_k) \right|,
\end{align*}
where $v_k\in L_k(r)$, and it is clear that $a_k^{(h)}$ does not depend on the choice of $v_k$.
For $k>h$, we set $a_k^{(h)}:=0$.
Then, 
\begin{align*}
    \Phi(T_{\Delta,h},\lambda)
    =1+\sum_{k=1}^{\infty} \sum_{v\in L_k(r)} \left| \Psi_{\mu_{T_{\Delta,h},\lambda}}(r,v) \right|
    =1+\sum_{k=1}^{\infty} \sum_{v\in L_k(r)}a_k^{(h)}
    =1+\sum_{k=1}^{\infty} (d+1)d^{k-1}a_k^{(h)}.
\end{align*}
When $k\leq h$, by \cref{fact:prod}, we have
\begin{align*}
    a_k^{(h)}=\prod_{j=1}^kp^{(h)}_{v_j}=\prod_{j=1}^kF_{d,\lambda}^{(h-j+1)}(0),
\end{align*}
where $r,v_1,\cdots,v_k$ is the unique path from $r$ to $v_k$ in $T_{\Delta,h}$, and the last equality follows from \cref{eq:p_k^h}.
Letting $h\rightarrow\infty$, it follows that  
\begin{align} \label{eq:limak}
    \lim_{h\rightarrow\infty} a_k^{(h)}=\prod_{j=1}^k \lim_{h\rightarrow\infty}F_{d,\lambda}^{(h-j+1)}(0)=\hat x^k,
\end{align}
where the last equality follows from the fact that $\lim_{h\rightarrow\infty}F_{d,\lambda}^{(h)}(0)=\hat x$ when $\lambda \le \lambda_c(\Delta)$ (see \cite{kelly1985stochastic,li2013correlation}).
Since $C_{\Delta,h,\lambda}\geq \Phi(T_{\Delta,h},\lambda)$, it holds that
\begin{align*}
    \liminf_{h\rightarrow\infty} C_{\Delta,h,\lambda}\geq\liminf_{h\rightarrow\infty} \Phi(T_{\Delta,h},\lambda)
    &=1+\liminf_{h\rightarrow\infty}\sum_{k=1}^{\infty} (d+1)d^{k-1}a_k^{(h)}\nonumber\\
    &\geq1+\sum_{k=1}^{\infty}\liminf_{h\rightarrow\infty}(d+1)d^{k-1}a_k^{(h)}\tag{by Fatou's lemma}\\
    &=1+ \sum_{k=1}^{\infty} (d+1)d^{k-1}\hat x^k\tag{by \cref{eq:limak}}\\
    &=1 + \frac{(d+1)\hat x}{1-d\hat x}=\frac{1+\hat x}{1-d\hat x}\tag{by $\hat x\leq \frac 1d$}.
\end{align*}
\cref{thm:SI-subcritical} shows that $C_{\Delta,h,\lambda}\leq\frac{1+\hat x}{1-d\hat x}$. Combined with the lower bound above, we have
\begin{align*}
    \lim_{h\rightarrow\infty} C_{\Delta,h,\lambda}=\frac{1+\hat x}{1-d\hat x},
\end{align*}
as desired.
\end{proof}

\section*{Acknowledgments}
The authors are grateful to Xiaoyu Chen and Xinyuan Zhang for their helpful and stimulating discussions.

\bibliographystyle{alpha}
\bibliography{ref.bib}

\appendix

\section{Proof of \texorpdfstring{\cref{prop:SS-mixing-full}}{Proposition 2.3}}
\label{subsec:pf-SI-mixing}

To prove \cref{prop:SS-mixing-full}, we need the following proposition, which can be concluded from Lemma 2.3, Lemma 3.13, Theorem 3.16, Lemma 3.17, and Lemma 3.19 of \cite{CCYZ24}.
\begin{proposition}[\cite{CCYZ24,chen2022localization}]\label{prop:SI-to-AT}
	Let $G$ be a $n$-vertex graph with maximum degree $\Delta\geq3$.
	Let $\mu_{G,\lambda}$ be the Gibbs distribution of the hardcore model on $G$ with fugacity $\lambda=\lambda_c(\Delta)$.
	Let $K:[0,1]\to \R_{\ge 0}$ be a function such that for every $\delta \in [0,1]$, the Gibbs distribution $\mu_{G,(1-\delta)\lambda}$
	satisfies spectral independence with constant $K(\delta)$.
	For any $\theta\in[0,1]$, if $(1-\theta)\lambda\leq\frac{1}{2\Delta}$, then the mixing time of Glauber dynamics satisfies
	\begin{align*}
		T_{\mathrm{mix}}(P_{\mathrm{GD}}) \leq 
		Cn^2\exp\tp{\int_0^{\theta}\frac{K(\delta)}{1-\delta}\-d\delta}\log\Delta,
	\end{align*}
	where $C>0$ is a constant independent of $n$. 
\end{proposition}

We now present the proof of \cref{prop:SS-mixing-full}.
\begin{proof}[Proof of \cref{prop:SS-mixing-full}]
	For every $\delta \in [0,1]$, $\mu_{G,(1-\delta)\lambda}$ satisfies spectral independence with constant $K(\delta)$, where $K(\delta)=\min\set{\frac{\rho}{\delta},n}$.
    The bound $K(\delta)\leq n$ holds because $\norm{\Psi_{\mu_{G,(1-\delta)\lambda}}}_\infty\leq n$.
	
	\noindent When $\frac{\rho}{n}\leq \theta$, i.e., $n \geq \frac{\rho}{\theta}$, it holds that 
	\begin{align*}\label{eq:expSI}
		\exp\tp{\int_0^{\theta}\frac{K(\delta)}{1-\delta}\-d\delta}
		&\leq \exp\tp{n\int_0^{\frac{\rho}{n}}\frac{1}{1-\delta}\-d\delta
			+\int_{\frac{\rho}{n}}^{\theta}\frac{\rho}{\delta}\frac{1}{1-\delta}\-d\delta}\\
        &=\exp\tp{n\log\tp{\frac{n}{n-\rho}}+\rho\log\tp{\frac{\theta}{1-\theta}}+\rho\log\tp{\frac{n-\rho}{\rho}}}\\
        &\leq \exp\tp{n\tp{\frac{n}{n-\rho}-1}+\rho\log\tp{\frac{\theta}{1-\theta}}+\rho\log\tp{\frac{n}{\rho}}}\\
        &\leq \exp\tp{\frac{1}{1-\theta}\rho+\rho\log\tp{\frac{\theta}{1-\theta}}+\rho\log\tp{\frac{n}{\rho}}}.
    \end{align*}
	Let $\theta=\frac{23}{24}$, when $n \geq \frac{\rho}{\theta}$, it holds that 
	\begin{align*}
		\exp\tp{\int_0^{\theta}\frac{K(\delta)}{1-\delta}\-d\delta}
		&\leq \exp \tp{24 \rho+\rho\log 23+ \rho \log\tp{\frac{n}{\rho}}}\\
		&=\exp\tp{24 \rho+\rho\log 23-\rho\log \rho}n^{\rho}=O(n^{\rho}).
	\end{align*}
	Notice that $(1-\theta)\lambda_c(\Delta)\leq\frac{1}{24}\cdot\frac{12}{\Delta}=\frac{1}{2\Delta}$, since $\lambda_c(\Delta)\leq\frac{12}{\Delta}$ for all $\Delta\geq 3$. 
	Then, for $n$ sufficiently large, specifically $n \geq \frac{\rho}{\theta}=\frac{24}{23}\rho$, by \cref{prop:SI-to-AT}, we have
	\begin{align*}
		T_{\mathrm{mix}}(P_{\mathrm{GD}}) &\leq 
		Cn^2\exp\tp{\int_0^{\theta}\frac{K(\delta)}{1-\delta}\-d\delta}\log\Delta=O\tp{n^{2+\rho}\log\Delta},
	\end{align*}
	as desired.
\end{proof}

\end{document}